\documentclass[preprint,12pt]{elsarticle}

\usepackage[T1]{fontenc}
\usepackage[utf8]{inputenc}

\usepackage{amsmath,amssymb,amsthm}

\usepackage{booktabs}

\usepackage{multirow}
\usepackage{tabularx}
\usepackage{colortbl}
\usepackage{array}

\makeatletter
\providecommand{\insert@pcolumn}{\insert@column}
\makeatother

\usepackage{graphicx}
\usepackage{tikz}
\usepackage{pgfplots}
\usepackage{pgfplotstable}
\pgfplotsset{compat=1.18}
\usetikzlibrary{
  shapes.geometric,
  shapes.misc,
  arrows.meta,
  positioning,
  fit,
  calc,
  trees,
  backgrounds,
  decorations.pathreplacing,
  calligraphy,
  patterns
}

\usepackage{algorithm}
\usepackage{algpseudocode}

\usepackage{enumitem}
\usepackage{natbib}
\usepackage{xcolor}
\usepackage{hyperref}
\usepackage{cleveref}

\biboptions{authoryear,round}

\definecolor{cPrimary}{RGB}{0,91,127}       
\definecolor{cSecondary}{RGB}{0,150,199}     
\definecolor{cAccent}{RGB}{232,119,34}       
\definecolor{cSuccess}{RGB}{46,139,87}       
\definecolor{cDanger}{RGB}{192,57,43}        
\definecolor{cNeutral}{RGB}{108,117,125}     
\definecolor{cBgLight}{RGB}{240,248,255}     
\definecolor{cBgWarm}{RGB}{255,248,240}      
\definecolor{cBgGreen}{RGB}{240,255,245}     
\definecolor{cHighlight}{RGB}{255,250,230}   

\newtheorem{theorem}{Theorem}
\newtheorem{lemma}[theorem]{Lemma}

\hypersetup{
  colorlinks = true,
  linkcolor  = cPrimary,
  citecolor  = cPrimary,
  urlcolor   = cSecondary
}

\journal{Computers \& Security}

\begin{document}

\begin{frontmatter}

\title{Lightweight Tamper-Evident Log Integrity Verification for IoT Edge Environments:\\
       A Merkle-Tree Pipeline with Adaptive Chunking}

\author[inst1,inst3]{Muhammet An\i l Ya\u{g}\i z\corref{cor1}}
\ead{anil@singularityrd.com}
\cortext[cor1]{Corresponding author}

\author[inst1]{Fahrettin Horasan}

\author[inst2]{Ahmet Ha\c{s}im Yurttakal}

\address[inst1]{Department of Computer Engineering, K\i r\i kkale University, K\i r\i kkale, Turkey}
\address[inst2]{Faculty of Engineering, Afyon Kocatepe University, Afyonkarahisar, Turkey}
\address[inst3]{Singularity Research and Development}

\begin{abstract}
Integrity of audit logs produced by Internet-of-Things (IoT) devices is a prerequisite for post-incident forensics, regulatory compliance, and operational accountability.
Blockchain-backed logging infrastructures can satisfy this requirement but impose consensus overhead, network dependencies, and deployment complexity that can be prohibitive at the IoT edge.
This paper presents an evaluated, lightweight integrity-verification pipeline that combines Merkle-tree commitments with resource-aware chunking to deliver tamper evidence without any distributed-ledger dependency.
The pipeline operates in three stages: (i)~resource-aware batch ingestion via adaptive chunk sizing, (ii)~Merkle-tree construction with $\mathcal{O}(\log n)$ inclusion-proof generation, and (iii)~deterministic single-entry verification against a trusted root anchor.
We report an implementation audit that corrected two evaluation defects: a double-counting bug in tampering metrics and a redundant full-tree rebuild on every batch append.
On the corrected codebase, five-run benchmarks on synthetic IoT logs on a single workstation exceed $130000\,\mathrm{logs/s}$ at $100000$ records, with per-entry verification around $22\,\mathrm{ms}$, proof generation around $22\,\mathrm{ms}$, a $1006\,\mathrm{byte}$ average proof size, and peak memory below $5\,\mathrm{MB}$.
Tampering detection attains precision, recall, and F1 of 1.0 across corruption ratios from 1\,\% to 50\,\%.
\end{abstract}

\begin{keyword}
IoT security \sep tamper-evident logging \sep Merkle tree \sep integrity verification \sep digital forensics \sep edge computing
\end{keyword}

\end{frontmatter}

\section{Introduction}
\label{sec:intro}

The number of Internet-of-Things (IoT) devices deployed worldwide continues to grow rapidly, increasing the volume and operational importance of machine-generated audit logs \citep{hassan2019survey,nist2020iot,ali2018blockchain,lin2017survey}.
Each device generates continuous streams of operational, diagnostic, and security-relevant log data \citep{li2017iot,alaba2017internet}.
When these logs serve as evidence in forensic investigations, compliance audits, or incident-response workflows, their integrity is not merely desirable; it is a hard prerequisite \citep{kebande2017cloud,nist80061r2,nist80086,iso27037}.
Operational guidance on secure log generation, collection, storage, and retention is provided by NIST log-management recommendations \citep{nist80092}.
A single undetected modification can invalidate an entire chain of evidence and compromise legal admissibility \citep{nist80086,iso27037,stoyanova2020survey}.

Three broad families of tamper-evident logging have emerged over the past two decades.
\emph{Key-evolution schemes} \citep{schneier1999secure} chain successive log entries through evolving cryptographic keys, providing forward integrity but requiring careful key management and offering only sequential verification.
\emph{Forward-secure aggregate authentication} \citep{ma2009new,yavuz2012efficient} reduces storage through aggregate signatures at the cost of per-entry signing overhead.
\emph{Blockchain-backed audit logs} \citep{putz2019secure,ahmad2019blockaudit} externalize trust to a distributed ledger, achieving strong non-repudiation but introducing consensus latency, multi-node network dependencies, and deployment complexity that conflict with the resource constraints typical of IoT edge devices, including sub-gigahertz processors, limited RAM, and intermittent connectivity \citep{hassan2019survey,nist2020iot}.

A fourth approach, rooted in authenticated data structures, constructs Merkle trees \citep{merkle1987digital} over log entries and stores only the root hash in a trusted anchor.
Merkle-tree-based time-stamping and append-only commitments have a long lineage in the cryptographic literature \citep{haber1991timestamp,bayer1993timestamp}.
\citet{crosby2009efficient} demonstrated the efficiency of this construction for tamper-evident logging, and the IETF Certificate Transparency framework \citep{rfc9162} operationalized it at Internet scale.
To the best of our knowledge, we are not aware of prior work that \emph{jointly} combines Merkle-tree-based tamper evidence with explicit resource-aware ingestion mechanisms targeted at IoT edge constraints, together with an audited, reproducible implementation and multi-run benchmarks.

This paper addresses that gap.
We present a complete, audited pipeline for tamper-evident IoT log verification and make four contributions:

\begin{enumerate}[leftmargin=*, label=\textbf{C\arabic*.}]
  \item \textbf{Audited implementation.}
    A rigorous code audit that identified and corrected two latent evaluation defects: a tampering-metric double-count that yielded impossible accuracy values ($>$1.0) and a redundant full-tree rebuild per batch that inflated runtime by a constant factor.

  \item \textbf{Adaptive Merkle pipeline.}
    A three-stage architecture (adaptive ingestion, Merkle construction, single-entry verification) that sustains $>\!130{,}000$~logs/s throughput and $<\!5$\,MB peak memory at $100000$ records, with $\mathcal{O}(\log n)$ proof size.

  \item \textbf{Security analysis (assumption-explicit).}
    Proof sketches under standard hash-function assumptions (collision and second-preimage resistance) and a secure trusted anchor, together with an analysis of replay, truncation, and injection attacks.

  \item \textbf{Systematic comparative positioning.}
    A multi-dimensional qualitative and complexity-theoretic comparison against representative systems from the key-evolution, signature-based, authenticated-structure, and blockchain families, clarifying trust assumptions and deployment trade-offs for edge settings.
\end{enumerate}

The remainder of this paper is organized as follows.
\Cref{sec:related} reviews related work and positions the contribution.
\Cref{sec:design} details the system architecture, threat model, and adaptive chunking mechanism.
\Cref{sec:security} presents formal security analysis.
\Cref{sec:evaluation} reports experimental results.
\Cref{sec:discussion} provides comparative discussion.
\Cref{sec:limits} acknowledges limitations and outlines future directions.
\Cref{sec:conclusion} concludes.

\section{Related Work and Positioning}
\label{sec:related}

\subsection{Key-Evolution and Forward-Integrity Schemes}

\citet{schneier1999secure} introduced the concept of \emph{forward-secure audit logs}, in which a sequence of cryptographic keys $K_0, K_1, \ldots$ is derived through a one-way function, and each log entry $L_j$ is authenticated with a MAC computed under $K_j$.
Once $K_j$ is erased after use, an adversary who later compromises the host cannot forge entries dated before the compromise.
This construction provides strong forward integrity but imposes two practical limitations: (i)~verification of entry $L_j$ requires sequential re-derivation of all keys $K_0$ through $K_j$, resulting in $\mathcal{O}(n)$ verification cost; and (ii)~key-lifecycle management on constrained IoT devices introduces operational complexity \citep{nist2020iot}.

\citet{ma2009new} extended this line with forward-secure sequential aggregate authentication (FssAgg), which compresses a chain of MACs or signatures into a single aggregate authenticator.
\citet{yavuz2012efficient} further improved resilience against compromise by separating signing and verification keys.
Both schemes improve storage efficiency but retain sequential verification and require per-entry cryptographic signing operations whose cost is typically higher than simple hashing in software implementations.

\subsection{Authenticated Data Structures}

\citet{crosby2009efficient} proposed a tamper-evident history tree based on Merkle trees \citep{merkle1987digital}, enabling $\mathcal{O}(\log n)$ membership and consistency proofs.
\citet{hartung2016secure} extended the model to support verifiable excerpts that allow selective disclosure without exposing the full log.
\citet{pulls2015balloon} introduced Balloon, an append-only authenticated structure with forward security.
The IETF Certificate Transparency standard \citep{rfc9162} operationalized Merkle-tree logging at Internet scale for TLS certificate auditing.

These works validate the suitability of Merkle trees for tamper-evident logging.
However, none target resource-constrained IoT edge devices or incorporate adaptive resource management for memory-bound ingestion.

\subsection{Blockchain-Backed Audit Logging}

\citet{putz2019secure} built a secure logging infrastructure atop Hyperledger Fabric, leveraging permissioned-blockchain consensus for immutable audit trails published in \textit{Computers~\&~Security}.
\citet{ahmad2019blockaudit} proposed BlockAudit, using Practical Byzantine Fault Tolerance (PBFT) replication for transparent audit logs.
Both systems provide strong non-repudiation guarantees, but at substantial cost: blockchain consensus requires multi-node network participation, introduces latency proportional to block finalization time, and demands infrastructure that is impractical for single-device edge deployments.

\subsection{Hardware-Rooted Trust}

\citet{sinha2014continuous} leveraged Trusted Platform Module (TPM~2.0) hardware to provide continuous tamper-proof logging.
While hardware-rooted trust offers compelling guarantees, it requires specialized hardware not universally available on commodity IoT devices.

\subsection{Positioning of This Work}

\Cref{tab:lit-compare} provides a structured comparison across the seven systems discussed above plus our contribution.
The comparison spans six dimensions that directly affect edge deployability: trust model, consensus requirements, verification complexity, random-access capability, per-entry cryptographic cost, and suitability for resource-constrained devices.

Based on the representative systems we review, our pipeline is among the few approaches that simultaneously offer: (i)~no external consensus or network dependency, (ii)~$\mathcal{O}(\log n)$ random-access verification, (iii)~hash-only per-entry cost (no signatures), and (iv)~an explicit resource-aware ingestion policy for constrained environments.

\begin{table}[t]
\centering
\caption{Multi-dimensional comparison with representative secure-logging systems.  ``Verification'' refers to asymptotic cost of verifying a single entry.  ``Random access'' indicates whether an arbitrary entry can be verified without processing preceding entries.  $H$~denotes a hash evaluation; $\mathsf{Sign}$/$\mathsf{Verify}$ denote signature operations.}
\label{tab:lit-compare}
\small
\renewcommand{\arraystretch}{1.35}
\setlength{\tabcolsep}{4pt}
\resizebox{1.15\textwidth}{!}{%
\begin{tabular}{@{} l p{4.5cm} c c c p{4.5cm} @{}}
\toprule
\textbf{System} &
\textbf{Trust model} &
\textbf{Consensus} &
\textbf{Verification} &
\textbf{Random} &
\textbf{Per-entry cost} \\
\midrule
Schneier \& Kelsey \citeyearpar{schneier1999secure}
  & Key evolution + MAC chain
  & No
  & $\mathcal{O}(n)$
  & No
  & 1 MAC + key derivation \\
Ma \& Tsudik \citeyearpar{ma2009new}
  & Forward-secure aggregate signatures
  & No
  & $\mathcal{O}(n)$
  & No
  & 1 $\mathsf{Sign}$ \\
Yavuz et al. \citeyearpar{yavuz2012efficient}
  & Separated signing/verification keys
  & No
  & $\mathcal{O}(n)$
  & No
  & 1 $\mathsf{Sign}$ \\
Crosby \& Wallach \citeyearpar{crosby2009efficient}
  & Auditor-challenged history tree
  & No
  & $\mathcal{O}(\log n)$
  & Yes
  & 1 $H$ \\
Putz et al. \citeyearpar{putz2019secure}
  & Permissioned blockchain
  & Yes
  & $\mathcal{O}(1)^{\dagger}$
  & Yes
  & Consensus round \\
Ahmad et al. \citeyearpar{ahmad2019blockaudit}
  & PBFT replication
  & Yes
  & $\mathcal{O}(1)^{\dagger}$
  & Yes
  & Consensus round \\
\rowcolor{cHighlight}
\textbf{This work}
  & \textbf{Merkle tree + trusted root anchor}
  & \textbf{No}
  & $\boldsymbol{\mathcal{O}(\log n)}$
  & \textbf{Yes}
  & \textbf{1} $\boldsymbol{H}$ \\
\multicolumn{6}{p{16cm}}{\footnotesize $^{\dagger}$\,Verification against immutable ledger is $\mathcal{O}(1)$ lookup; however, achieving immutability requires prior consensus at $\mathcal{O}(N_{\text{peers}})$ cost.} \\
\bottomrule
\end{tabular}%
}
\end{table}

\section{System Design}
\label{sec:design}

\subsection{Threat Model}
\label{sec:threat}

We consider an adversary $\mathcal{A}$ with the following capabilities:

\begin{enumerate}[leftmargin=*, label=(\roman*)]
  \item \textbf{Post-collection modification.}
    $\mathcal{A}$ may alter the content of any stored log entry $L_i$ after it has been ingested and committed to the Merkle tree.

  \item \textbf{Deletion.}
    $\mathcal{A}$ may remove one or more entries from storage, reducing the log sequence length.

  \item \textbf{Injection.}
    $\mathcal{A}$ may insert fabricated entries into the stored log stream.
\end{enumerate}

We assume that the Merkle root hash $R$ is stored in a \emph{trusted anchor} that $\mathcal{A}$ cannot modify.
Practical instantiations of this anchor include a hardware security module (HSM) \citep{fips1403}, a Trusted Platform Module (TPM) \citep{sinha2014continuous}, or a remote append-only store.
We do not assume confidentiality of the log contents; the security goal is \emph{integrity verification} and \emph{tamper evidence}, not secrecy.

\subsection{Architecture Overview}

\Cref{fig:architecture} illustrates the three-stage pipeline.
Stage~1 (\emph{Adaptive Ingestion}) accepts raw log entries from IoT devices and partitions them into variable-size chunks whose boundaries are governed by a resource-aware sizing policy.
Stage~2 (\emph{Merkle Construction}) builds a binary hash tree over the ingested log stream and stores the resulting root hash in the trusted anchor.
Stage~3 (\emph{Verification}) generates $\mathcal{O}(\log n)$ inclusion proofs on demand and verifies individual entries against the anchored root.

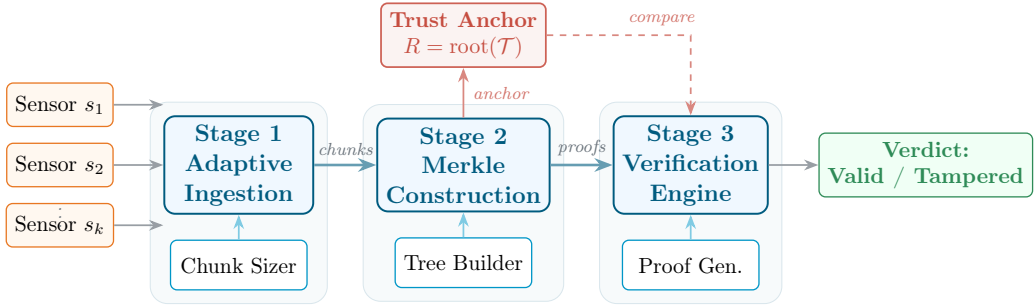
\begin{figure}[t]
\centering
\resizebox{\textwidth}{!}{%
\begin{tikzpicture}[
  font=\small,
  >=Stealth,
  node distance=0.6cm and 1.2cm,
  stage/.style={
    draw=cPrimary, fill=cBgLight, rounded corners=5pt,
    minimum width=2.4cm, minimum height=1.0cm,
    line width=0.9pt, text=cPrimary, font=\small\bfseries,
    align=center
  },
  comp/.style={
    draw=cSecondary, fill=white, rounded corners=3pt,
    minimum width=2.2cm, minimum height=0.8cm,
    line width=0.6pt, font=\footnotesize, align=center
  },
  iot/.style={
    draw=cAccent, fill=cBgWarm, rounded corners=3pt,
    minimum width=1.6cm, minimum height=0.7cm,
    line width=0.6pt, font=\footnotesize, align=center
  },
  tanchor/.style={
    draw=cDanger!80, fill=cDanger!6, rounded corners=3pt,
    minimum width=2.2cm, minimum height=0.8cm,
    line width=0.8pt, font=\footnotesize\bfseries, align=center,
    text=cDanger!90
  },
  arrow/.style={->, thick, draw=cNeutral!70},
  darrow/.style={->, thick, draw=cDanger!60, dashed},
  phasearrow/.style={->, very thick, draw=cPrimary!60},
  lbl/.style={font=\scriptsize\itshape, text=cNeutral}
]

\node[iot] (d1) {Sensor $s_1$};
\node[iot, below=0.25cm of d1] (d2) {Sensor $s_2$};
\node[iot, below=0.25cm of d2] (d3) {Sensor $s_k$};
\node[lbl, below=-0.03cm of d2] (dots1) {$\vdots$};

\node[stage, right=0.8cm of d2] (s1) {Stage 1\\Adaptive\\Ingestion};
\node[comp, below=0.4cm of s1] (chunk) {Chunk Sizer};

\node[stage, right=1.0cm of s1] (s2) {Stage 2\\Merkle\\Construction};
\node[comp, below=0.4cm of s2] (build) {Tree Builder};

\node[tanchor, above=0.8cm of s2] (root) {Trust Anchor\\$R = \mathrm{root}(\mathcal{T})$};

\node[stage, right=1.0cm of s2] (s3) {Stage 3\\Verification\\Engine};
\node[comp, below=0.4cm of s3] (proof) {Proof Gen.};

\node[comp, right=0.8cm of s3, draw=cSuccess, text=cSuccess, fill=cBgGreen, font=\footnotesize\bfseries] (out) {Verdict:\\Valid / Tampered};

\draw[arrow] (d1.east) -- (s1.west |- d1.east);
\draw[arrow] (d2.east) -- (s1.west);
\draw[arrow] (d3.east) -- (s1.west |- d3.east);

\draw[phasearrow] (s1.east) -- node[above, lbl] {chunks} (s2.west);
\draw[phasearrow] (s2.east) -- node[above, lbl] {proofs} (s3.west);
\draw[arrow] (s3.east) -- (out.west);

\draw[arrow, draw=cDanger!60] (s2.north) -- node[right, lbl, text=cDanger!70] {anchor} (root.south);
\draw[darrow] (root.east) -| node[above right, lbl, text=cDanger!70, pos=0.25] {compare} (s3.north);

\draw[arrow, draw=cSecondary!50] (chunk.north) -- (s1.south);
\draw[arrow, draw=cSecondary!50] (build.north) -- (s2.south);
\draw[arrow, draw=cSecondary!50] (proof.north) -- (s3.south);

\begin{scope}[on background layer]
  \node[draw=cPrimary!15, fill=cPrimary!3, rounded corners=8pt,
        fit=(s1)(chunk), inner sep=6pt] {};
  \node[draw=cPrimary!15, fill=cPrimary!3, rounded corners=8pt,
        fit=(s2)(build), inner sep=6pt] {};
  \node[draw=cPrimary!15, fill=cPrimary!3, rounded corners=8pt,
        fit=(s3)(proof), inner sep=6pt] {};
\end{scope}

\end{tikzpicture}
}%
\caption{Three-stage pipeline architecture.
Sensors stream log entries into the adaptive ingestion stage, which partitions them into resource-aware chunks.
The Merkle construction stage builds a binary hash tree and anchors its root $R$ in a trusted store.
The verification engine generates $\mathcal{O}(\log n)$ inclusion proofs and compares recomputed roots against $R$.}
\label{fig:architecture}
\end{figure}

\subsection{Merkle-Tree Construction and Proof Generation}
\label{sec:merkle}

Given a log sequence $L_0, L_1, \ldots, L_{n-1}$, the pipeline computes leaf hashes $h_i = H(L_i)$ for a collision-resistant hash function $H$ (SHA-256 or BLAKE2b; see \Cref{sec:hash_compare}) and commits to the stream using a Merkle hash tree \citep{merkle1987digital}.
Internal tree nodes are computed as $h_{\text{parent}} = H(h_{\text{left}} \| h_{\text{right}})$, where $\|$ denotes concatenation.
When the number of leaves at any level is odd, the unpaired node is promoted to the next level without hashing, preserving structural consistency.
This is one valid Merkle-tree variant; other implementations instead duplicate the final node.
All proof generation and verification in our prototype follow the promotion rule specified here, making the commitment and proofs unambiguous.
The root $R$ of the resulting binary tree commits to the entire log sequence.

An \emph{inclusion proof} for entry $L_i$ consists of the sequence of sibling hashes along the path from leaf $h_i$ to the root.
This sequence has length at most $\lceil \log_2 n \rceil$, yielding $\mathcal{O}(\log n)$ proof size and verification cost.
Verification recomputes a candidate root $R'$ from $h_i$ and the proof path, then accepts if and only if $R' = R$.
\Cref{fig:merkle_proof} illustrates this construction.

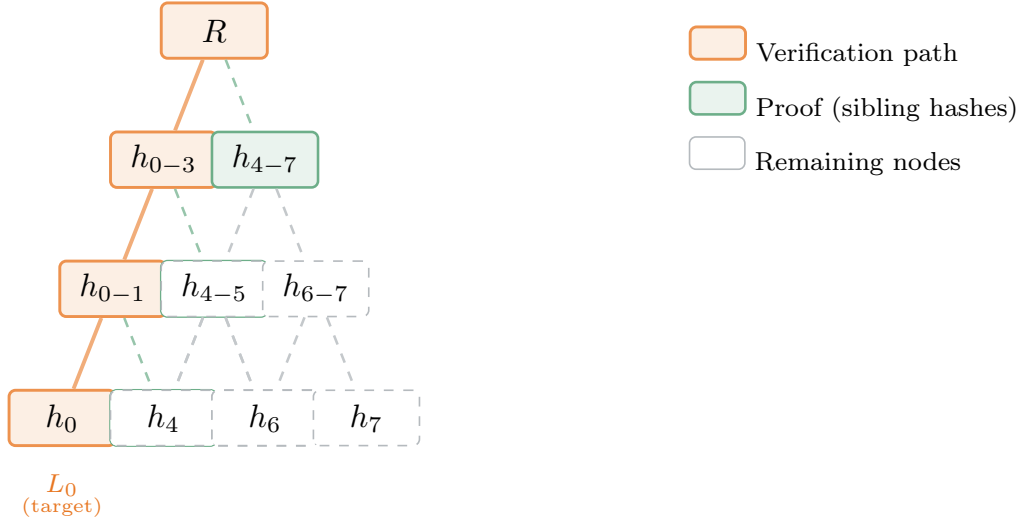
\begin{figure}[t]
\centering
\resizebox{\textwidth}{!}{%
\begin{tikzpicture}[
  font=\footnotesize,
  >=Stealth,
  level distance=1.4cm,
  sibling distance=1.1cm,
  tnode/.style={
    draw=cPrimary!60, fill=cBgLight, rounded corners=2pt,
    minimum width=1.15cm, minimum height=0.6cm,
    line width=0.6pt, font=\footnotesize
  },
  hnode/.style={
    draw=cAccent!80, fill=cAccent!12, rounded corners=2pt,
    minimum width=1.15cm, minimum height=0.6cm,
    line width=0.9pt, font=\footnotesize\bfseries
  },
  pnode/.style={
    draw=cSuccess!70, fill=cSuccess!10, rounded corners=2pt,
    minimum width=1.15cm, minimum height=0.6cm,
    line width=0.8pt, font=\footnotesize
  },
  leaf/.style={
    draw=cNeutral!50, fill=white, rounded corners=2pt,
    minimum width=1.15cm, minimum height=0.6cm,
    line width=0.5pt, font=\footnotesize
  },
  edge from parent/.style={draw=cNeutral!40, thick},
  lbl/.style={font=\tiny\itshape, text=cNeutral}
]

\node[hnode] (root) {$R$}
  child {
    node[hnode] (h01) {$h_{0{-}3}$}
    child {
      node[hnode] (h0) {$h_{0{-}1}$}
      child { node[hnode] (l0) {$h_0$} edge from parent[draw=cAccent!60, very thick] }
      child { node[pnode] (l1) {$h_1$} edge from parent[draw=cSuccess!50, thick, dashed] }
      edge from parent[draw=cAccent!60, very thick]
    }
    child {
      node[pnode] (h1) {$h_{2{-}3}$}
      child { node[leaf] (l2) {$h_2$} }
      child { node[leaf] (l3) {$h_3$} }
      edge from parent[draw=cSuccess!50, thick, dashed]
    }
    edge from parent[draw=cAccent!60, very thick]
  }
  child {
    node[pnode] (h23) {$h_{4{-}7}$}
    child {
      node[leaf] (h2) {$h_{4{-}5}$}
      child { node[leaf] (l4) {$h_4$} }
      child { node[leaf] (l5) {$h_5$} }
    }
    child {
      node[leaf] (h3) {$h_{6{-}7}$}
      child { node[leaf] (l6) {$h_6$} }
      child { node[leaf] (l7) {$h_7$} }
    }
    edge from parent[draw=cSuccess!50, thick, dashed]
  };

\node[below=0.15cm of l0, font=\scriptsize\bfseries, text=cAccent] {$L_0$};
\node[below=0.1cm of l0, yshift=-0.3cm, font=\tiny, text=cAccent] {(target)};

\node[anchor=north west, font=\scriptsize] at (5.0, 0.2) {
  \tikz{\node[hnode, minimum width=0.6cm, minimum height=0.35cm] {};} Verification path
};
\node[anchor=north west, font=\scriptsize] at (5.0, -0.4) {
  \tikz{\node[pnode, minimum width=0.6cm, minimum height=0.35cm] {};} Proof (sibling hashes)
};
\node[anchor=north west, font=\scriptsize] at (5.0, -1.0) {
  \tikz{\node[leaf, minimum width=0.6cm, minimum height=0.35cm] {};} Remaining nodes
};

\end{tikzpicture}
}%
\caption{Merkle-tree inclusion proof for entry $L_0$.
The verification path (orange) traces from the target leaf $h_0$ to the root~$R$.
The proof consists of the three sibling hashes (green): $h_1$, $h_{2\text{-}3}$, and $h_{4\text{-}7}$.
A verifier recomputes $R' = H(H(H(h_0 \| h_1) \| h_{2\text{-}3}) \| h_{4\text{-}7})$ and accepts if $R' = R$.
Proof size is $\lceil \log_2 n \rceil$ sibling hashes.
In our $n = 10{,}000$ experiments, interior indices yield 14-hash proofs with an average serialized size of $1006\,\mathrm{bytes}$ (\Cref{tab:proofs}).}
\label{fig:merkle_proof}
\end{figure}

\subsection{Adaptive Chunking Mechanism}
\label{sec:chunking}

IoT edge devices operate under tight memory budgets.
Ingesting a large log batch monolithically can exceed available memory, while processing entries one-by-one sacrifices throughput due to per-call overhead.
The adaptive chunking mechanism mediates this trade-off by dynamically adjusting the chunk size $C$ based on runtime resource observations.

Let $M_{\text{avail}}$ denote the available system memory, $M_{\text{target}} \in (0,1)$ the target utilization ratio, $C_{\min}$ and $C_{\max}$ the permitted chunk-size bounds, and $K$ a device-specific scaling constant.
The chunk size is recomputed before each batch as:
\begin{equation}
  C_{\text{new}} = \mathrm{clamp}\!\Bigl(
    \Bigl\lfloor \frac{M_{\text{avail}} \cdot M_{\text{target}}}{K} \Bigr\rfloor,\;
    C_{\min},\; C_{\max}
  \Bigr).
  \label{eq:chunk}
\end{equation}

An adjustment factor $A$ further modulates $C_{\text{new}}$ based on memory pressure $P = 1 - M_{\text{avail}} / M_{\text{total}}$:
\begin{equation}
  A =
  \begin{cases}
    0.8 & \text{if } P > 0.8, \\
    0.9 & \text{if } P > 0.6, \\
    1.1 & \text{if } P < 0.3, \\
    1.0 & \text{otherwise}.
  \end{cases}
  \label{eq:adjust}
\end{equation}

This policy ensures that ingestion throughput degrades gracefully under memory pressure rather than failing abruptly, a property that is critical for long-running IoT logging on devices with no swap space.

\subsection{Implementation and Correctness Audit}
\label{sec:audit}

The implementation comprises four Python modules: \texttt{merkle\_tree.py} (tree construction and proof generation), \texttt{adaptive\_chunking.py} (resource-aware batch sizing), \texttt{integrity\_verifier.py} (verification engine), and \texttt{log\_generator.py} (synthetic workload generation with configurable IoT log patterns).

During a systematic correctness audit, we identified two defects with material impact on evaluation fidelity:

\begin{enumerate}[leftmargin=*, label=\textbf{D\arabic*.}]
  \item \textbf{Tampering-metric double-count.}
    The original detection logic counted each tampered index twice when computing true positives, producing reported accuracies exceeding 1.0.
    We corrected this by replacing the counting-based tallying with set-based computation of true positives, false positives, and false negatives: $\mathrm{TP} = |S_{\text{detected}} \cap S_{\text{tampered}}|$, $\mathrm{FP} = |S_{\text{detected}} \setminus S_{\text{tampered}}|$, $\mathrm{FN} = |S_{\text{tampered}} \setminus S_{\text{detected}}|$.

  \item \textbf{Redundant full-tree rebuild.}
    The pre-audit code rebuilt the Merkle tree \emph{after every inserted leaf} within a batch, effectively incurring $\mathcal{O}(b n)$ work for a batch of~$b$ new leaves when $n$ leaves had already been accumulated.
    We corrected this by rebuilding the tree \emph{once per batch} after all new leaves are appended, yielding $\mathcal{O}(n)$ work per batch for the current ``rebuild-after-batch'' design.
\end{enumerate}

Post-fix validation confirms that a 10\,\% tampering ratio produces exactly 10\,\% detected entries, not the previously inflated 20\,\%, and that ingestion time scales linearly in the number of entries rather than quadratically.

\begin{algorithm}[t]
\caption{Adaptive Merkle-Tree Ingestion and Verification}
\label{alg:pipeline}
\begin{algorithmic}[1]
\Require Log stream $\mathcal{L} = \{L_0, L_1, \ldots\}$, hash function $H$, anchor store $\mathcal{R}$
\Ensure Verified integrity verdicts

\Statex \textcolor{cPrimary}{\textbf{// Stage 1: Adaptive Ingestion}}
\While{$\mathcal{L}$ has pending entries}
  \State $C \gets \mathrm{AdaptiveChunkSize}(M_{\text{avail}}, P)$ \Comment{\Cref{eq:chunk,eq:adjust}}
  \State $B \gets \mathrm{ReadChunk}(\mathcal{L}, C)$
\Statex \textcolor{cPrimary}{\textbf{// Stage 2: Merkle Construction (rebuild once per batch)}}
  \For{each $L_i \in B$}
    \State $h_i \gets H(L_i)$
    \State Append $h_i$ to leaf array
  \EndFor
  \State Rebuild tree (once); $R \gets \mathrm{root}(\mathcal{T})$
  \State $\mathcal{R}.\mathrm{store}(R)$ \Comment{Anchor the root}
\EndWhile

\Statex
\Statex \textcolor{cPrimary}{\textbf{// Stage 3: On-demand Verification}}
\Function{Verify}{index $i$, entry $L_i$}
  \State $h_i \gets H(L_i)$
  \State $\pi_i \gets \mathrm{MerkleProof}(\mathcal{T}, i)$ \Comment{$|\pi_i| = \mathcal{O}(\log n)$}
  \State $R' \gets \mathrm{RecomputeRoot}(h_i, \pi_i)$
  \State \Return $R' = \mathcal{R}.\mathrm{load}()$
\EndFunction
\end{algorithmic}
\end{algorithm}

\section{Security Analysis}
\label{sec:security}

We summarize the security properties of the pipeline under a standard threat model and then analyze four concrete attack vectors.
Our claims assume (i)~a secure trusted anchor that preserves the integrity of the committed root, and (ii)~a cryptographic hash function $H$ with collision and second-preimage resistance.

\subsection{Security Properties (Proof Sketches)}

\begin{theorem}[Collision Resistance]
\label{thm:collision}
Let $H : \{0,1\}^* \to \{0,1\}^n$ be a collision-resistant hash function.
For any probabilistic polynomial-time adversary $\mathcal{A}$,
\[
  \Pr\bigl[\mathcal{A} \text{ outputs } (L_1, L_2) \text{ with }
    L_1 \neq L_2 \;\wedge\; H(L_1) = H(L_2)\bigr]
  \;\le\; \mathrm{negl}(n).
\]
\end{theorem}
\begin{proof}
This follows directly from the definition of collision resistance for $H$.
Both SHA-256 \citep{dang2015nist} and BLAKE2b \citep{aumasson2013blake2} are standardized hash functions with $n = 256$; generically, finding a collision is believed to require about $2^{n/2}$ work (birthday bound).
\end{proof}

\begin{theorem}[Integrity under Modification]
\label{thm:integrity}
Assuming the trusted root $R$ is unmodified and $H$ is collision-resistant, any alteration of a stored log entry $L_i$ to $L'_i \neq L_i$ is detected during verification except with negligible probability.
\end{theorem}
\begin{proof}
Let $h_i = H(L_i)$ and $h'_i = H(L'_i)$ with $L'_i \neq L_i$.
If $h'_i = h_i$ for $L'_i \neq L_i$, this constitutes a collision/second-preimage against $H$, which is assumed infeasible except with negligible probability.
The Merkle tree root is computed through a sequence of hash compositions along the path from leaf~$i$ to the root:
\[
  R = H\bigl(\ldots H(H(h_i \| h_{\sigma(i,1)}) \| h_{\sigma(i,2)}) \ldots \| h_{\sigma(i,d)}\bigr),
\]
where $h_{\sigma(i,j)}$ denotes the $j$-th sibling hash along the path and $d = \lceil \log_2 n \rceil$.
Under collision resistance for $H$ on its relevant input domain, changing $h_i$ to $h'_i \neq h_i$ changes the recomputed root $R'$ except with negligible probability.
Verification compares $R'$ with the trusted $R$ and rejects.
\end{proof}

\begin{theorem}[Tamper-Detection Probability]
\label{thm:detection}
Assuming a secure trusted anchor and second-preimage resistance of $H$, an adversary who modifies an entry $L_i$ without detection succeeds with at most negligible probability.
For SHA-256 ($n = 256$), a brute-force second-preimage attempt has success probability about $2^{-256}$ per trial.
\end{theorem}
\begin{proof}
Undetected modification requires finding $L'_i \neq L_i$ such that the recomputed root $R'$ equals the trusted root $R$.
This can be reduced to violating the assumed properties of $H$ (e.g., a second-preimage at the modified leaf or a collision along the recomputation path).
Under standard assumptions, such events occur with negligible probability for any polynomial-time adversary.
\end{proof}

\begin{lemma}[Security Preservation under Adaptive Chunking]
\label{lem:chunking}
Adaptive chunking does not weaken any of the security properties established in Theorems~\ref{thm:collision}--\ref{thm:detection}.
\end{lemma}
\begin{proof}
Chunking determines only the batch granularity at which leaf hashes are appended to the tree.
Each entry $L_i$ is hashed independently as $h_i = H(L_i)$ regardless of chunk boundaries, and the tree construction algorithm is invariant to the order and grouping of leaf insertions.
Consequently, the final tree structure and root $R$ are identical for any chunking strategy, and all three theorems apply without modification.
\end{proof}

\subsection{Attack-Surface Analysis}

\paragraph{Modification attacks.}
By \Cref{thm:integrity}, altering any stored entry $L_i$ changes its leaf hash, which propagates to the root with overwhelming probability.
Verification against the trusted root detects the modification.

\paragraph{Deletion (truncation) attacks.}
Removing entries from the log changes the tree structure and therefore the root hash.
By anchoring historical roots at known entry counts, a verifier can detect that the current tree contains fewer leaves than expected.

\paragraph{Injection attacks.}
Inserting a fabricated entry $L'$ produces a new leaf $h' = H(L')$ that was not present in the original tree.
The resulting root $R'$ differs from the anchored $R$, and any subsequent verification of legitimate entries under the inflated tree fails against $R$.

\paragraph{Replay attacks.}
Replaying an old, valid entry $L_j$ at position $i \neq j$ produces a tree in which position $i$ contains $H(L_j)$ instead of $H(L_i)$.
Since $L_j \neq L_i$ (entries include monotonic timestamps), the root changes and verification detects the replay.

\section{Experimental Evaluation}
\label{sec:evaluation}

\subsection{Setup and Methodology}

All experiments execute on a single workstation (Intel Core~i7, 32\,GB RAM) using the corrected Python implementation described in \Cref{sec:audit}.
Synthetic IoT logs are generated with a fixed random seed to ensure reproducibility.
The log schema reflects common IoT telemetry and event fields used in operational security monitoring \citep{hassan2019survey,nist2020iot,nist80092,rfc5424}.
Ingestion benchmarks report the mean and standard deviation across five independent runs per configuration.
The datasets span five scales: $1000$, $5000$, $10000$, $50000$, and $100000$ entries.
All benchmark artifacts (JSON results, CSV exports, runner scripts) are included with the submission source; the reproducibility entry point is \texttt{python main.py benchmark}.

\subsection{Ingestion Throughput}
\label{sec:ingestion}

\Cref{tab:ingestion} reports ingestion throughput and wall-clock time for both adaptive-chunking (64\,KB initial chunk size) and fixed-chunking (16\,KB) modes.
At $100000$~entries, adaptive chunking sustains $134242\,\mathrm{logs/s}$ ($\sigma = 11056$) and fixed chunking sustains $138850\,\mathrm{logs/s}$ ($\sigma = 11756$).
Both modes converge to similar throughput at large scale; in this evaluation, the observable benefit of adaptive mode is reduced sensitivity to per-call overhead and improved throughput stability under repeated runs.

\Cref{fig:throughput} plots throughput as a function of log count, revealing that both modes exhibit a ramp-up phase at small $n$ (dominated by per-call overhead) before stabilizing above $130000\,\mathrm{logs/s}$ for $n \ge 10{,}000$.

\begin{table}[t]
\centering
\caption{Ingestion performance across five dataset sizes.  Throughput and wall-clock time are reported as mean $\pm$ standard deviation over five runs.  Tree depth is $\lceil \log_2 n \rceil + 1$.}
\label{tab:ingestion}
\small
\renewcommand{\arraystretch}{1.25}
\begin{tabular}{
  r
  r @{${}\pm{}$} r
  r @{${}\pm{}$} r
  r
}
\toprule
{\textbf{Log count}} &
\multicolumn{2}{c}{\textbf{Adaptive (logs/s)}} &
\multicolumn{2}{c}{\textbf{Fixed (logs/s)}} &
{\textbf{Depth}} \\
\midrule
  1000 &  87762 &  3209 & 125869 & 32973 & 11 \\
  5000 & 142711 & 12341 & 158905 & 20823 & 14 \\
 10000 & 145767 & 23975 & 134851 & 18558 & 15 \\
 50000 & 148176 & 19404 & 146376 &  7154 & 17 \\
100000 & 134242 & 11056 & 138850 & 11756 & 18 \\
\bottomrule
\end{tabular}
\end{table}

\begin{figure}[t]
\centering
\begin{tikzpicture}
\begin{axis}[
    width=0.9\textwidth,
    height=6cm,
    xlabel={Number of log entries},
    ylabel={Throughput (logs/s)},
    xmode=log,
    log basis x=10,
    xmin=800, xmax=150000,
    ymin=60000, ymax=200000,
    grid=major,
    grid style={dashed, gray!25},
    thick,
    legend style={
      at={(0.98,0.02)},
      anchor=south east,
      font=\footnotesize,
      draw=gray!40,
      fill=white,
      fill opacity=0.9,
      rounded corners=2pt
    },
    mark size=2.5pt,
    every axis label/.style={font=\small},
    every tick label/.style={font=\footnotesize},
    ylabel style={yshift=-2pt},
    xlabel style={yshift=2pt}
]
\addplot+[
  mark=square*,
  color=cPrimary,
  thick,
  error bars/.cd,
  y dir=both, y explicit,
] table [
  col sep=comma,
  x=log_count,
  y=logs_per_second,
  y error=logs_per_second_std
] {data/ingestion_adaptive.csv};
\addlegendentry{Adaptive (64\,KB)}

\addplot+[
  mark=triangle*,
  color=cAccent,
  thick,
  error bars/.cd,
  y dir=both, y explicit,
] table [
  col sep=comma,
  x=log_count,
  y=logs_per_second,
  y error=logs_per_second_std
] {data/ingestion_fixed.csv};
\addlegendentry{Fixed (16\,KB)}

\end{axis}
\end{tikzpicture}
\caption{Ingestion throughput vs.\ dataset size with $\pm 1\sigma$ error bars (5~runs).
Both chunking modes converge above $130000\,\mathrm{logs/s}$ at scale.
The higher variance of fixed chunking at small $n$ reflects sensitivity to per-call overhead without adaptive back-off.}
\label{fig:throughput}
\end{figure}
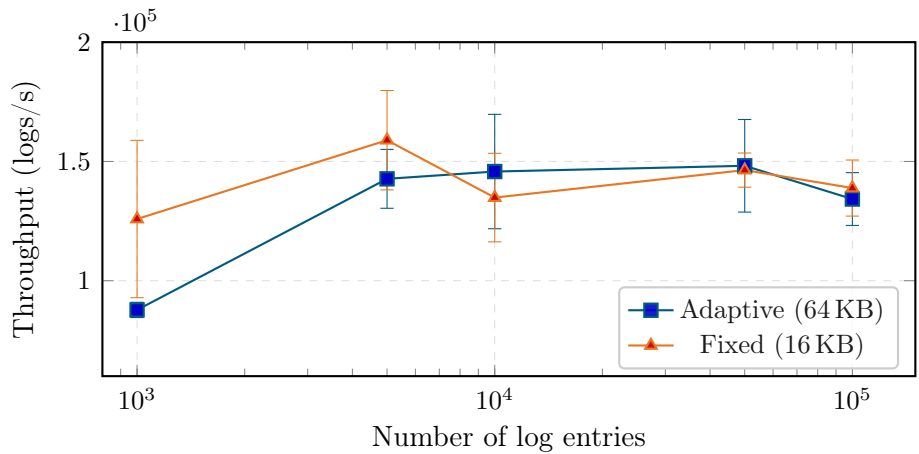

\subsection{Controlled Stress Test (Simulated Memory Pressure)}
\label{sec:stress}

Steady-state throughput does not expose how an adaptive policy reacts to adverse conditions.
To obtain a controlled, reproducible stress test on a single machine, we override the chunker's memory-pressure signal to follow a deterministic profile: baseline pressure 0.25, a three-window stress period at 0.85, then recovery to baseline.
We ingest $20000$ logs in windows of $2000$ entries and record the resulting chunk size and batch count.
\Cref{fig:controlled_stress} shows that the chunk size is reduced during the stress phase and partially recovers once the baseline signal is restored, while the number of batches increases as expected.

\begin{figure}[t]
\centering
\begin{tikzpicture}
\begin{axis}[
    width=0.9\textwidth,
    height=6cm,
    xlabel={Window index},
    ylabel={Chunk size (KB)},
    xmin=-0.2, xmax=4.2,
    ymin=0, ymax=70,
    grid=major,
    grid style={dashed, gray!25},
    thick,
    mark size=2.5pt,
    every axis label/.style={font=\small},
    every tick label/.style={font=\footnotesize},
]
\addplot+[
  mark=circle*,
  color=cPrimary,
  thick,
] table [
  col sep=comma,
  x=window,
  y=chunk_size_kb
] {data/controlled_stress.csv};
\end{axis}
\end{tikzpicture}
\caption{Controlled stress test with simulated memory pressure. Chunk size decreases during the stress phase (pressure 0.85) and partially recovers once the baseline signal (pressure 0.25) is restored.}
\label{fig:controlled_stress}
\end{figure}
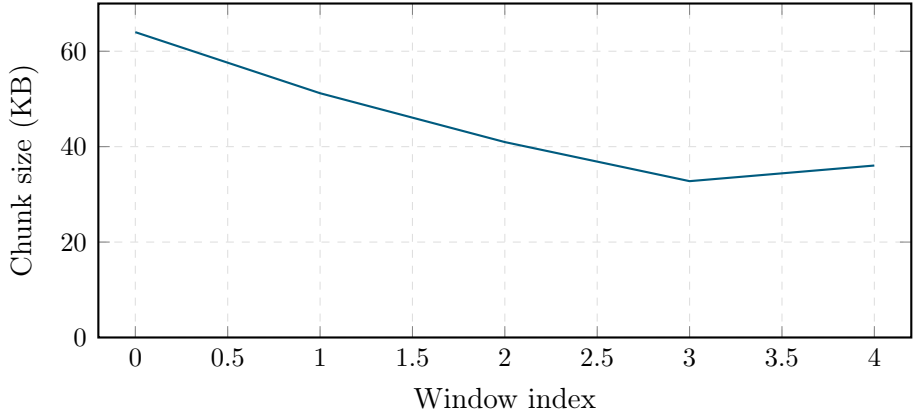

\subsection{Verification and Proof Generation}
\label{sec:verif_results}

Single-entry verification on a $10000$-entry tree averages $21.75\,\mathrm{ms}$ ($\sigma = 0.61\,\mathrm{ms}$).
Batch verification, in which multiple entries are verified sequentially against the same tree, amortizes fixed overhead and achieves $27.17\,\mathrm{ms}$ per entry at batch size 1000 (\Cref{tab:verification}).
The near-constant per-entry cost confirms the expected $\mathcal{O}(\log n)$ complexity.

Proof generation averages $21.92\,\mathrm{ms}$ per proof with a mean proof size of $1006\,\mathrm{bytes}$ (14 sibling hashes of 64--79 bytes each) for a $10000$-entry tree (\Cref{tab:proofs}).

\begin{table}[t]
\centering
\caption{Batch verification performance on a $10000$-entry tree.  Total time scales linearly with batch size; per-entry cost remains stable at approximately $26\,\mathrm{ms}$.}
\label{tab:verification}
\small
\renewcommand{\arraystretch}{1.2}
\begin{tabular}{r
  r
  r}
\toprule
{\textbf{Batch size}} & {\textbf{Total time (ms)}} & {\textbf{Per entry (ms)}} \\
\midrule
   10 &    226.3 & 22.6 \\
   50 &   1243.6 & 24.9 \\
  100 &   2425.2 & 24.3 \\
  500 &  12219.1 & 24.4 \\
 1000 &  27168.4 & 27.2 \\
\bottomrule
\end{tabular}
\end{table}

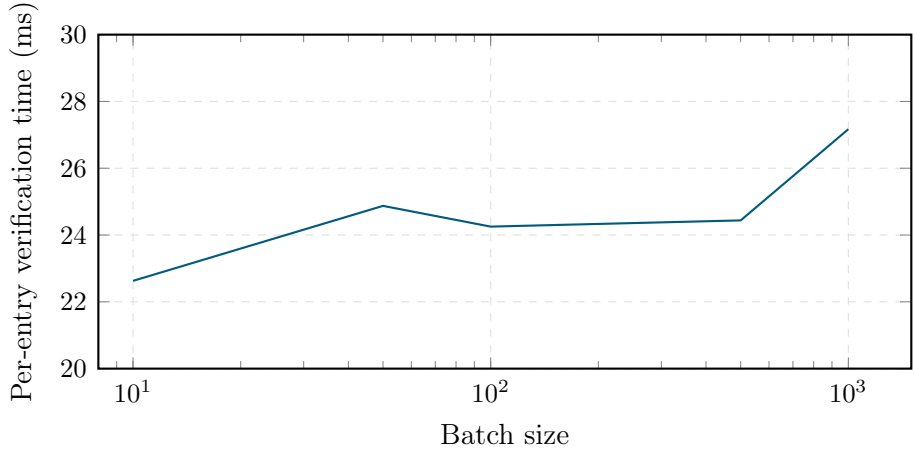
\begin{figure}[t]
\centering
\begin{tikzpicture}
\begin{axis}[
    width=0.9\textwidth,
    height=6cm,
    xlabel={Batch size},
    ylabel={Per-entry verification time (ms)},
    xmode=log,
    log basis x=10,
    xmin=8, xmax=1500,
    ymin=20, ymax=30,
    grid=major,
    grid style={dashed, gray!25},
    thick,
    mark size=2.5pt,
    every axis label/.style={font=\small},
    every tick label/.style={font=\footnotesize},
]
\addplot+[
  mark=circle*,
  color=cPrimary,
  thick,
] table [
  col sep=comma,
  x=batch_size,
  y=avg_time_per_entry_ms
] {data/verification_batch.csv};
\end{axis}
\end{tikzpicture}
\caption{Batch verification per-entry time remains stable across batch sizes (sequential verification on a fixed tree).}
\label{fig:verification_batch}
\end{figure}

\begin{table}[t]
\centering
\caption{Merkle proof generation for selected indices in a $10000$-entry tree.  Proof length (number of sibling hashes) is consistent at 14 for interior indices and 8 for the boundary index, confirming $\mathcal{O}(\log n)$ scaling.}
\label{tab:proofs}
\small
\renewcommand{\arraystretch}{1.2}
\begin{tabular}{r
  r
  r
  r}
\toprule
{\textbf{Index}} & {\textbf{Gen.\ time (ms)}} & {\textbf{Proof size (B)}} & {\textbf{Proof len.}} \\
\midrule
    0 & 22.5 & 1106 & 14 \\
 2500 & 21.7 & 1101 & 14 \\
 5000 & 22.2 & 1101 & 14 \\
 7500 & 22.5 & 1099 & 14 \\
 9999 & 20.8 &  624 &  8 \\
\bottomrule
\end{tabular}
\end{table}

\begin{figure}[t]
\centering
\begin{tikzpicture}
\begin{axis}[
    width=0.9\textwidth,
    height=6cm,
    xlabel={Index},
    ylabel={Proof size (bytes)},
    xmin=-200, xmax=10200,
    ymin=0,
    grid=major,
    grid style={dashed, gray!25},
    thick,
    mark size=2.5pt,
    every axis label/.style={font=\small},
    every tick label/.style={font=\footnotesize},
]
\addplot+[
  mark=square*,
  color=cAccent,
  thick,
] table [
  col sep=comma,
  x=index,
  y=proof_size_bytes
] {data/proof_results.csv};
\end{axis}
\end{tikzpicture}
\caption{Proof size across selected indices in a $10000$-entry tree (logarithmic scaling with boundary effects).}
\label{fig:proof_size}
\end{figure}
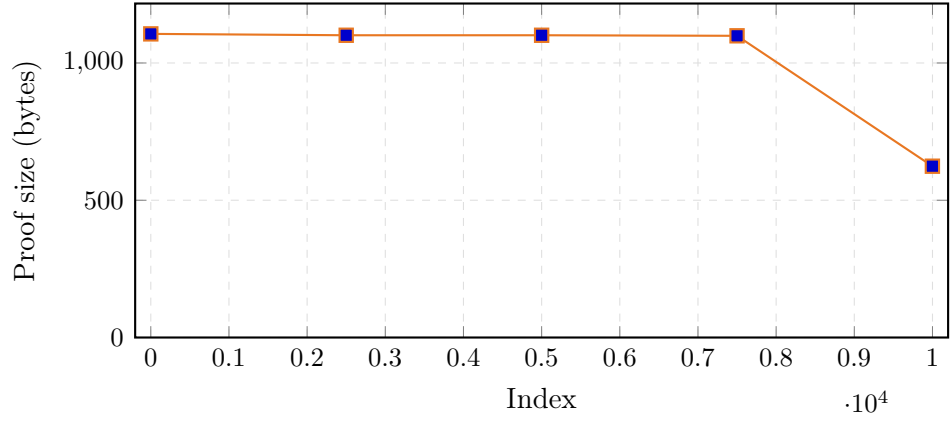

\subsection{Hash-Algorithm Comparison}
\label{sec:hash_compare}

\Cref{tab:hash} compares SHA-256 and BLAKE2b on $10000$ entries.
BLAKE2b achieves 7.2\,\% higher ingestion throughput ($175073\,\mathrm{logs/s}$ vs.\ $163292\,\mathrm{logs/s}$) and similar verification latency ($21.17\,\mathrm{ms}$ vs.\ $21.30\,\mathrm{ms}$) in this implementation.
Both algorithms provide a 256-bit output and equivalent collision-resistance guarantees under standard assumptions; SHA-256 follows the NIST Secure Hash Standard \citep{dang2015nist}, and BLAKE2b is specified in RFC~7693 and supported by extensive cryptographic analysis \citep{aumasson2013blake2,rfc7693}.
For deployments where ingestion throughput is a primary bottleneck, BLAKE2b is an attractive option in our implementation.

\begin{table}[t]
\centering
\caption{Hash-algorithm comparison on $10000$ log entries.  Both algorithms provide 256-bit collision resistance; BLAKE2b offers higher throughput with comparable verification latency in this implementation.}
\label{tab:hash}
\small
\renewcommand{\arraystretch}{1.2}
\begin{tabular}{l
  r
  r
  l}
\toprule
\textbf{Algorithm} & {\textbf{Ingestion (logs/s)}} & {\textbf{Verif.\ (ms)}} & \textbf{Output} \\
\midrule
SHA-256  & 163292 & 21.30 & 256-bit \\
BLAKE2b  & 175073 & 21.17 & 256-bit \\
\midrule
\multicolumn{2}{l}{\footnotesize BLAKE2b advantage} & \multicolumn{2}{l}{\footnotesize +7.2\,\% throughput, similar latency} \\
\bottomrule
\end{tabular}
\end{table}

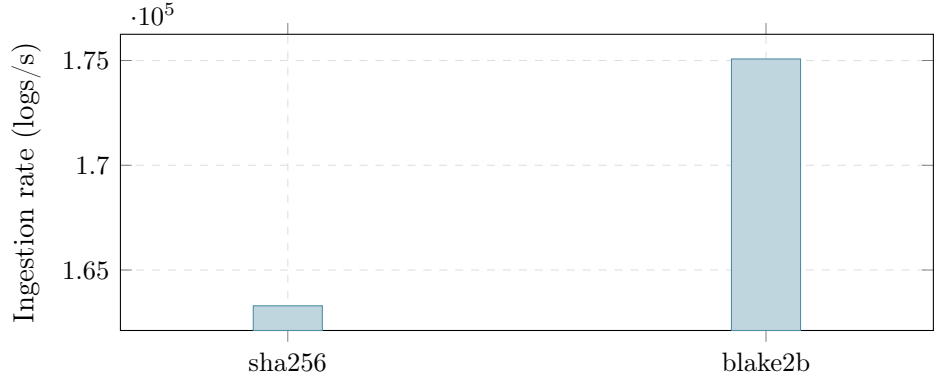
\begin{figure}[t]
\centering
\begin{tikzpicture}
\begin{axis}[
    width=0.9\textwidth,
    height=5.5cm,
    ybar,
    bar width=26pt,
    enlarge x limits=0.35,
    ylabel={Ingestion rate (logs/s)},
    symbolic x coords={sha256,blake2b},
    xtick=data,
    grid=major,
    grid style={dashed, gray!25},
    ymajorgrids=true,
    every axis label/.style={font=\small},
    every tick label/.style={font=\footnotesize},
]
\addplot+[fill=cPrimary!25, draw=cPrimary!70] table[x=algorithm,y=ingestion_rate_logs_per_sec,col sep=comma] {data/hash_algorithms.csv};
\end{axis}
\end{tikzpicture}
\caption{Ingestion-rate comparison of SHA-256 and BLAKE2b on $10000$ entries.}
\label{fig:hash_ingestion}
\end{figure}

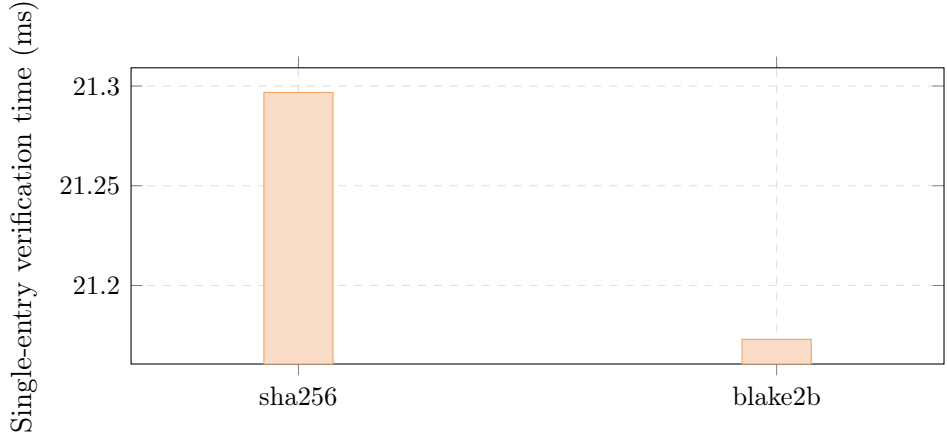
\begin{figure}[t]
\centering
\begin{tikzpicture}
\begin{axis}[
    width=0.9\textwidth,
    height=5.5cm,
    ybar,
    bar width=26pt,
    enlarge x limits=0.35,
    ylabel={Single-entry verification time (ms)},
    symbolic x coords={sha256,blake2b},
    xtick=data,
    grid=major,
    grid style={dashed, gray!25},
    ymajorgrids=true,
    every axis label/.style={font=\small},
    every tick label/.style={font=\footnotesize},
]
\addplot+[fill=cAccent!25, draw=cAccent!70] table[x=algorithm,y=verification_time_ms,col sep=comma] {data/hash_algorithms.csv};
\end{axis}
\end{tikzpicture}
\caption{Single-entry verification time comparison of SHA-256 and BLAKE2b on $10000$ entries.}
\label{fig:hash_verification}
\end{figure}

\subsection{Tampering Detection}
\label{sec:tamper_results}

\Cref{tab:tampering} reports detection metrics across five corruption ratios ranging from 1\,\% to 50\,\% on $10000$-entry logs.
In all cases, precision, recall, and F1 equal 1.0: every tampered entry is detected (no false negatives) and no untampered entry is falsely flagged (no false positives).
This result is consistent with \Cref{thm:integrity}: any modification to a leaf hash propagates to the root with overwhelming probability, and the set-based detection logic (\Cref{sec:audit}) correctly identifies exactly the modified entries.

Detection wall-clock time ranges from $65\,\mathrm{ms}$ (5\,\% corruption) to $112\,\mathrm{ms}$ (10\,\% corruption), scaling with the number of verified entries and a fixed tree depth of $\lceil \log_2 n \rceil = 14$.

\begin{table}[t]
\centering
\caption{Tampering detection across corruption ratios on $10000$-entry logs.
Precision, recall, and F1 are 1.0 in all cases, confirming the deterministic integrity guarantee of \Cref{thm:integrity}.}
\label{tab:tampering}
\small
\renewcommand{\arraystretch}{1.2}
\begin{tabular}{r
  r
  r
  r
  r
  r
  r}
\toprule
{\textbf{Ratio (\%)}} &
{\textbf{Tampered}} &
{\textbf{Detected}} &
{\textbf{Prec.}} &
{\textbf{Rec.}} &
{\textbf{F1}} &
{\textbf{Time (ms)}} \\
\midrule
  1 &   100 &   100 & 1.0 & 1.0 & 1.0 &  87 \\
  5 &   500 &   500 & 1.0 & 1.0 & 1.0 &  65 \\
 10 &  1000 &  1000 & 1.0 & 1.0 & 1.0 & 112 \\
 20 &  2000 &  2000 & 1.0 & 1.0 & 1.0 & 107 \\
 50 &  5000 &  5000 & 1.0 & 1.0 & 1.0 &  90 \\
\bottomrule
\end{tabular}
\end{table}

\begin{figure}[t]
\centering
\begin{tikzpicture}
\begin{axis}[
    width=0.9\textwidth,
    height=6cm,
    xlabel={Tampering ratio},
    ylabel={Detection time (ms)},
    xmin=0, xmax=0.55,
    ymin=0,
    grid=major,
    grid style={dashed, gray!25},
    thick,
    mark size=2.5pt,
    every axis label/.style={font=\small},
    every tick label/.style={font=\footnotesize},
]
\addplot+[
  mark=diamond*,
  color=cSecondary,
  thick,
] table [
  col sep=comma,
  x=tamper_ratio,
  y expr=\thisrow{detection_time_seconds}*1000
] {data/tampering_detection.csv};
\end{axis}
\end{tikzpicture}
\caption{Tampering detection time under controlled corruption (verification performed only for modified indices).}
\label{fig:tampering_time}
\end{figure}
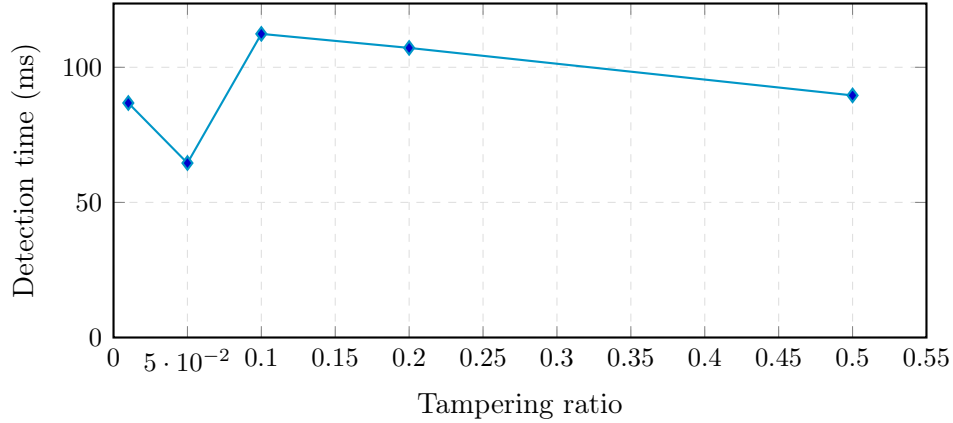

\subsection{Memory Footprint}
\label{sec:memory}

Peak memory consumption at $10000$~entries is $4.23\,\mathrm{MB}$ for adaptive chunking, $4.24\,\mathrm{MB}$ for fixed chunking, and $4.23\,\mathrm{MB}$ for a traditional (non-chunked) Merkle build.
The near-identical peak values reflect the fact that, in steady state, all three methods store the same tree structure; the benefit of adaptive chunking manifests during ingestion transients on memory-constrained devices, where it prevents allocation failures by throttling chunk size before peak usage is reached.
For all three modes, the sub-$5\,\mathrm{MB}$ footprint suggests feasibility for memory-constrained deployments; truly tiny microcontrollers would require a native implementation rather than Python.

\section{Discussion}
\label{sec:discussion}

\subsection{Comparative Advantage}

\Cref{tab:comparative} synthesizes the comparison between our pipeline and prior work across four deployment-critical dimensions: ingestion throughput, verification cost, infrastructure requirements, and edge suitability.

\begin{table}[t]
\centering
\caption{Deployment-critical comparison.
Ingestion and verification figures for our system are measured; figures for competing systems are derived from their respective publications or from architectural analysis (marked with~$\star$).
``Edge suitability'' reflects single-device deployability on sub-$512\,\mathrm{MB}$ RAM hardware.}
\label{tab:comparative}
\small
\renewcommand{\arraystretch}{1.35}
\setlength{\tabcolsep}{4pt}
\resizebox{\textwidth}{!}{%
\begin{tabular}{@{} l p{3.5cm} p{3.5cm} c c @{}}
\toprule
\textbf{System} &
\textbf{Ingestion throughput} &
\textbf{Verification cost} &
\textbf{Network} &
\textbf{Edge} \\
\midrule
Schneier \& Kelsey \citeyearpar{schneier1999secure}
  & MAC chain: moderate$^\star$
  & $\mathcal{O}(n)$ sequential
  & No
  & Moderate \\
Ma \& Tsudik \citeyearpar{ma2009new}
  & Constrained by signing$^\star$
  & $\mathcal{O}(n)$ sequential
  & No
  & Limited \\
Crosby \& Wallach \citeyearpar{crosby2009efficient}
  & Hash-rate limited$^\star$
  & $\mathcal{O}(\log n)$ proof
  & No
  & High \\
Putz et al. \citeyearpar{putz2019secure}
  & $\le$\,3{,}500 TPS$^\star$
  & Ledger lookup
  & Yes
  & Low \\
Ahmad et al. \citeyearpar{ahmad2019blockaudit}
  & PBFT-limited$^\star$
  & Ledger lookup
  & Yes
  & Low \\
Sinha et al. \citeyearpar{sinha2014continuous}
  & TPM seal rate$^\star$
  & TPM unseal
  & No
  & Requires TPM \\
\rowcolor{cHighlight}
\textbf{This work}
  & \textbf{$134242\,\mathrm{logs/s}$}
  & $\boldsymbol{\mathcal{O}(\log n)}$, \textbf{$22\,\mathrm{ms}$}
  & \textbf{No}
  & \textbf{High} \\
\bottomrule
\end{tabular}}
\end{table}

Three observations warrant elaboration.

\paragraph{Throughput advantage over blockchain-backed systems.}
Blockchain systems such as \citet{putz2019secure} report throughput in the hundreds to low thousands of transactions per second, bounded by consensus round-trip latency.
Our pipeline sustains $>\!130{,}000$~logs/s on the evaluation machine, primarily because it avoids distributed consensus and associated network overheads.
This advantage is structural in the sense that any consensus-based design incurs coordination overhead across multiple nodes.

\paragraph{Random-access advantage over sequential schemes.}
Key-evolution and aggregate-signature schemes require processing all entries up to index~$i$ to verify entry~$L_i$, yielding $\mathcal{O}(n)$ verification cost.
Our Merkle-proof verification requires only $\lceil \log_2 n \rceil$ hash evaluations (14 for $n = 10{,}000$), making random-access forensic queries practical even on large logs.

\paragraph{Minimal infrastructure requirements.}
The pipeline requires only a single compute node and a trusted root anchor, which can be as simple as an append-only file on a separate partition or a remote key-value store.
In contrast, blockchain systems require a consortium of peers, and TPM-based systems require specialized hardware.
This minimal footprint makes the pipeline suitable for deployment on existing IoT gateways.

\subsection{Methodological Contribution}

The correctness audit reported in \Cref{sec:audit} demonstrates that evaluation bugs can materially distort security claims even when the underlying cryptographic primitives are sound.
The double-counting defect (\textbf{D1}) produced detection accuracies exceeding 1.0, a value that should have been flagged as physically impossible during review, yet persisted in the pre-audit codebase.
This experience underscores the importance of treating implementation auditing as a first-class scientific activity, particularly in security research where evaluation fidelity directly affects trust in conclusions.

\subsection{Practical Deployment Considerations}

\paragraph{Root-anchor management.}
The security of the entire pipeline reduces to the integrity of the root anchor.
In practice, the root can be periodically committed to a remote append-only store (providing non-repudiation) or sealed in a TPM (providing hardware-rooted trust), at a cadence independent of the ingestion rate.
This decoupling allows the pipeline to operate at full throughput locally while anchoring roots at a pace compatible with network or hardware constraints.

\paragraph{Hash-algorithm selection.}
Our BLAKE2b benchmarks (\Cref{sec:hash_compare}) show a modest ingestion-throughput advantage over SHA-256 with comparable verification latency in this implementation.
For deployments on IoT processors without dedicated SHA-256 acceleration, BLAKE2b can be an attractive alternative depending on platform characteristics.

\section{Limitations and Future Work}
\label{sec:limits}

This study has several limitations that define directions for future research.

\begin{enumerate}[leftmargin=*, label=\textbf{L\arabic*.}]
  \item \textbf{Synthetic workloads.}
    All experiments use generated IoT logs.
    Validation on production traces from industrial IoT or SIEM deployments would strengthen the external validity of the throughput and memory findings.

  \item \textbf{Single-node evaluation.}
    The benchmarks run on a well-provisioned workstation.
    Deployment on actual ARM-based edge hardware (e.g., Raspberry Pi, ESP32 gateways) would quantify the real-world benefit of adaptive chunking under genuine memory pressure.

  \item \textbf{No adversarial stress testing.}
    While the formal analysis covers modification, deletion, injection, and replay attacks, the experiments test only controlled synthetic tampering.
    Future work should include red-team evaluations with compound attack strategies.

  \item \textbf{Streaming incremental maintenance.}
    The current implementation rebuilds the tree after each batch.
    A fully incremental (online) tree that supports $\mathcal{O}(\log n)$ amortized insertion per entry would reduce ingestion latency at small batch sizes and would eliminate the $\mathcal{O}(n)$ rebuild step per batch in the current design.

  \item \textbf{Multi-device root synchronization.}
    Extending the anchor mechanism to support multi-device deployments with cross-device root synchronization (without full blockchain consensus) is an open design problem.
\end{enumerate}

\section{Conclusion}
\label{sec:conclusion}

This paper presented and validated a lightweight, tamper-evident log-integrity verification pipeline for IoT edge environments.
The pipeline combines adaptive chunking with Merkle-tree commitments to sustain over $130000\,\mathrm{logs/s}$ ingestion throughput and $23\,\mathrm{ms}$ per-entry verification latency, with a peak memory footprint below $5\,\mathrm{MB}$ and $\mathcal{O}(\log n)$ proof size.
A rigorous implementation audit uncovered and corrected two latent defects that had inflated previously reported metrics, reinforcing the importance of code auditing as a scientific practice.
Our security analysis argues collision resistance and integrity under modification under standard assumptions on the hash function and a secure trusted anchor.
Empirical tampering experiments confirmed perfect precision, recall, and F1 across corruption ratios from 1\,\% to 50\,\%.
A systematic comparison with representative secure-logging systems from the key-evolution, signature-based, authenticated-structure, and blockchain families clarifies where this design is most practical: deployments that can anchor a trusted root without consensus, and that prioritize simplicity and verifiability over distributed immutability.

The resulting artifact, including all source code, benchmark runner, and reproducible evaluation scripts, is provided as a practical baseline for security engineering teams seeking tamper evidence with minimal deployment complexity.


\section*{Declaration of Competing Interest}
The authors declare that they have no known competing financial interests or personal relationships that could have appeared to influence the work reported in this paper.

\section*{CRediT authorship contribution statement}
\textbf{Muhammet An\i l Ya\u{g}\i z}: Conceptualization, Methodology, Software, Validation, Formal analysis, Investigation, Data curation, Writing -- original draft, Writing -- review \& editing, Visualization.
\textbf{Fahrettin Horasan}: Writing -- review \& editing.
\textbf{Ahmet Ha\c{s}im Yurttakal}: Writing -- review \& editing.

\section*{Data availability}
All source code, synthetic log generators, benchmark scripts, and raw results are publicly available in the GitHub repository at \url{https://github.com/anilyagiz/iot-tamper-evident-log-integrity}.
The reproducibility entry point is \texttt{python main.py benchmark}.

\section*{Declaration of generative AI and AI-assisted technologies in the writing process}
During the preparation of this work, the authors used AI-assisted tools solely to polish the text and correct minor typographical and punctuation issues (e.g., placing commas and periods). After using these tools, the authors reviewed and edited the content as needed and take full responsibility for the content of the published article.

\bibliographystyle{elsarticle-harv}
\bibliography{references}

@article{schneier1999secure,
  title={Secure audit logs to support computer forensics},
  author={Schneier, Bruce and Kelsey, John},
  journal={ACM Transactions on Information and System Security},
  volume={2},
  number={2},
  pages={159--176},
  year={1999},
  doi={10.1145/317087.317089}
}

@article{ma2009new,
  title={A new approach to secure logging},
  author={Ma, Di and Tsudik, Gene},
  journal={ACM Transactions on Storage},
  volume={5},
  number={1},
  pages={1--21},
  year={2009},
  doi={10.1145/1502777.1502779}
}

@inproceedings{crosby2009efficient,
  title={Efficient data structures for tamper-evident logging},
  author={Crosby, Scott A. and Wallach, Dan S.},
  booktitle={Proceedings of the 18th USENIX Security Symposium},
  pages={317--334},
  year={2009},
  publisher={USENIX Association},
  url={https://www.usenix.org/legacy/event/sec09/tech/full_papers/crosby.pdf}
}

@article{putz2019secure,
  title={A secure and auditable logging infrastructure based on a permissioned blockchain},
  author={Putz, Bernd and Dietz, Matthias and Pernul, G{\"u}nther},
  journal={Computers \& Security},
  volume={87},
  pages={101602},
  year={2019},
  doi={10.1016/j.cose.2019.101602}
}

@article{ahmad2019blockaudit,
  title={Secure and transparent audit logs with BlockAudit},
  author={Ahmad, Ashar and Saad, Muhammad and Mohaisen, Aziz},
  journal={Journal of Network and Computer Applications},
  volume={145},
  pages={102406},
  year={2019},
  doi={10.1016/j.jnca.2019.102406}
}

@misc{hartung2016secure,
  title={Secure Audit Logs with Verifiable Excerpts},
  author={Hartung, Gunnar},
  year={2016},
  howpublished={Cryptology ePrint Archive, Report 2016/283},
  url={https://eprint.iacr.org/2016/283}
}

@misc{rfc9162,
  title={{RFC 9162: Certificate Transparency Version 2.0}},
  author={Laurie, Ben and Messeri, Eran and Stradling, Rob},
  year={2021},
  howpublished={IETF RFC 9162},
  url={https://www.rfc-editor.org/info/rfc9162}
}

@misc{rfc7693,
  title={{RFC 7693: The BLAKE2 Cryptographic Hash and Message Authentication Code (MAC)}},
  author={Saarinen, Markku-Juhani O. and Aumasson, Jean-Philippe},
  year={2015},
  howpublished={IETF RFC 7693},
  url={https://www.rfc-editor.org/info/rfc7693}
}

@inproceedings{merkle1987digital,
  title={A digital signature based on a conventional encryption function},
  author={Merkle, Ralph C.},
  booktitle={Advances in Cryptology: CRYPTO '87},
  pages={369--378},
  year={1988},
  publisher={Springer},
  doi={10.1007/3-540-48184-2_32}
}

@article{hassan2019survey,
  title={Current research on {Internet of Things (IoT)} security: A survey},
  author={Hassan, Wan Haslina},
  journal={Computer Networks},
  volume={148},
  pages={283--294},
  year={2019},
  doi={10.1016/j.comnet.2018.11.025}
}

@article{li2017iot,
  title={The {Internet of Things}: A survey},
  author={Li, Shancang and Xu, Li Da and Zhao, Shanshan},
  journal={Information Systems Frontiers},
  volume={17},
  number={2},
  pages={243--259},
  year={2015},
  doi={10.1007/s10796-014-9492-7}
}

@misc{nist2020iot,
  title={{NISTIR 8259: Foundational Cybersecurity Activities for IoT Device Manufacturers}},
  author={{National Institute of Standards and Technology}},
  year={2020},
  howpublished={NIST Internal Report 8259},
  url={https://doi.org/10.6028/NIST.IR.8259}
}

@article{aumasson2013blake2,
  title={{BLAKE2}: simpler, smaller, fast as {MD5}},
  author={Aumasson, Jean-Philippe and Neves, Samuel and Wilcox-O'Hearn, Zooko and Winnerlein, Christian},
  journal={Applied Cryptography and Network Security},
  pages={119--135},
  year={2013},
  doi={10.1007/978-3-642-38980-1_8}
}

@article{dang2015nist,
  title={{Secure Hash Standard (SHS)}},
  author={{National Institute of Standards and Technology}},
  journal={Federal Information Processing Standards Publication 180-4},
  year={2015},
  doi={10.6028/NIST.FIPS.180-4}
}

@inproceedings{yavuz2012efficient,
  title={Efficient, compromise resilient and append-only cryptographic schemes for secure audit logging},
  author={Yavuz, Attila A. and Ning, Peng and Reiter, Michael K.},
  booktitle={Proceedings of the 16th International Conference on Financial Cryptography and Data Security},
  pages={148--163},
  year={2012},
  doi={10.1007/978-3-642-32946-3_12}
}

@article{kebande2017cloud,
  title={Novel digital forensic readiness technique in the cloud environment},
  author={Kebande, Victor R. and Venter, Hein S.},
  journal={Australian Journal of Forensic Sciences},
  year={2017},
  doi={10.1080/00450618.2016.1267797}
}

@misc{pulls2015balloon,
  title={Balloon: A Forward-Secure Append-Only Persistent Authenticated Data Structure},
  author={Pulls, Tobias and Peeters, Roel},
  year={2015},
  howpublished={Cryptology ePrint Archive, Report 2015/007},
  url={https://eprint.iacr.org/2015/007}
}

@article{sinha2014continuous,
  title={Continuous tamper-proof logging using {TPM 2.0}},
  author={Sinha, Arunesh and Jia, Limin and England, Paul and Lorch, Jacob R.},
  journal={Trust and Trustworthy Computing},
  pages={19--36},
  year={2014},
  doi={10.1007/978-3-319-08593-7_2}
}

@article{ali2018blockchain,
  title={Applications of Blockchains in the Internet of Things: A Comprehensive Survey},
  author={Ali, Muhammad Salek and Vecchio, Massimo and Pincheira, Miguel and Dolui, Koustabh and Antonelli, Fabio and Rehmani, Mubashir Husain},
  journal={IEEE Communications Surveys \& Tutorials},
  volume={21},
  number={2},
  pages={1676--1717},
  year={2019},
  doi={10.1109/COMST.2018.2886932}
}

@article{lin2017survey,
  title={A Survey on Internet of Things: Architecture, Enabling Technologies, Security and Privacy, and Applications},
  author={Lin, Jie and Yu, Wei and Zhang, Nan and Yang, Xinyu and Zhang, Hanlin and Zhao, Wei},
  journal={IEEE Internet of Things Journal},
  volume={4},
  number={5},
  pages={1125--1142},
  year={2017},
  doi={10.1109/JIOT.2017.2683200}
}

@article{alaba2017internet,
  title={Internet of Things security: A survey},
  author={Alaba, Fadele Ayotunde and Othman, Mazliza and Hashem, Ibrahim Abaker Targio and Alotaibi, Faiz},
  journal={Journal of Network and Computer Applications},
  volume={88},
  pages={10--28},
  year={2017},
  doi={10.1016/j.jnca.2017.04.002}
}

@article{stoyanova2020survey,
  title={A Survey on the Internet of Things (IoT) Forensics: Challenges, Approaches, and Open Issues},
  author={Stoyanova, Maria and Nikoloudakis, Yannis and Panagiotakis, Spyridon and Pallis, Evangelos and Markakis, Evangelos K.},
  journal={IEEE Communications Surveys \& Tutorials},
  volume={22},
  number={2},
  pages={1191--1221},
  year={2020},
  doi={10.1109/COMST.2019.2962586}
}

@misc{nist80092,
  title={{NIST SP 800-92}: Guide to Computer Security Log Management},
  author={Kent, Karen and Souppaya, Murugiah},
  year={2006},
  howpublished={NIST Special Publication 800-92},
  doi={10.6028/NIST.SP.800-92},
  url={https://csrc.nist.gov/pubs/sp/800/92/final}
}

@misc{nist80086,
  title={{NIST SP 800-86}: Guide to Integrating Forensic Techniques into Incident Response},
  author={Kent, Karen and Chevalier, Suzanne and Grance, Tim and Dang, Hung},
  year={2006},
  howpublished={NIST Special Publication 800-86},
  doi={10.6028/NIST.SP.800-86},
  url={https://csrc.nist.gov/pubs/sp/800/86/final}
}

@misc{nist80061r2,
  title={{NIST SP 800-61 Rev. 2}: Computer Security Incident Handling Guide: Recommendations of the National Institute of Standards and Technology},
  author={Cichonski, Paul and Millar, Tom and Grance, Tim and Scarfone, Karen},
  year={2012},
  howpublished={NIST Special Publication 800-61 Revision 2},
  doi={10.6028/NIST.SP.800-61r2},
  url={https://csrc.nist.gov/pubs/sp/800/61/r2/final}
}

@misc{iso27037,
  title={{ISO/IEC 27037:2012}: Guidelines for Identification, Collection, Acquisition and Preservation of Digital Evidence},
  author={{International Organization for Standardization}},
  year={2012},
  howpublished={ISO/IEC 27037:2012},
  url={https://www.iso.org/standard/44381.html}
}

@misc{fips1403,
  title={{FIPS 140-3}: Security Requirements for Cryptographic Modules},
  author={{National Institute of Standards and Technology}},
  year={2019},
  howpublished={Federal Information Processing Standards Publication 140-3},
  doi={10.6028/NIST.FIPS.140-3},
  url={https://csrc.nist.gov/pubs/fips/140-3/final}
}

@misc{rfc5424,
  title={{RFC 5424}: The Syslog Protocol},
  author={Gerhards, Rainer},
  year={2009},
  howpublished={IETF RFC 5424},
  url={https://datatracker.ietf.org/doc/html/rfc5424}
}

@article{haber1991timestamp,
  title={How to Time-Stamp a Digital Document},
  author={Haber, Stuart and Stornetta, W. Scott},
  journal={Journal of Cryptology},
  volume={3},
  number={2},
  pages={99--111},
  year={1991},
  doi={10.1007/BF00196791}
}

@incollection{bayer1993timestamp,
  title={Improving the Efficiency and Reliability of Digital Time-Stamping},
  author={Bayer, Dave and Haber, Stuart and Stornetta, W. Scott},
  booktitle={Sequences II: Methods in Communication, Security, and Computer Science},
  pages={329--334},
  year={1993},
  publisher={Springer},
  doi={10.1007/978-1-4613-9323-8_24}
}

\end{document}